\documentclass[format=acmsmall,review=false,screen=true,authorversion=true]{acmart}
\pdfoutput=1

\setcopyright{none}
\acmDOI{}

\usepackage{multirow}
\usepackage[utf8]{inputenc}
\usepackage[T1]{fontenc}
\usepackage[australian,american]{babel}
\usepackage{graphicx}
\usepackage{thmtools, mathdots}
\usepackage[framemethod=tikz]{mdframed}
\usepackage{pgfplots}
\pgfplotsset{compat=1.12}
\usepackage{array}
\usepackage[binary-units]{siunitx}
\usepackage{booktabs}
\usepackage{multirow}
\usepackage{pifont}
\usepackage{tikz}
\usepackage{csquotes}
\setcitestyle{square,aysep={},yysep={;}}
\usepackage{color}
\usepackage{todonotes}
\usepackage{lipsum}
\usepackage{subfigure}
\usepackage{pgfplotstable}

\newcommand{\E}{\mathbb{E}}
\renewcommand{\P}{\mathbb{P}}
\renewcommand{\epsilon}{\varepsilon} 
\renewcommand{\phi}{\varphi}
\renewcommand{\emph}{\textbf}

\newcommand{\fanin}{f_\text{in}}
\newcommand{\fanout}{f_\text{out}}

\newcommand{\uninformed}[1]{u_{#1}}
\newcommand{\informed}[1]{i_{#1}}

\usepackage[ruled]{algorithm2e}

\newboolean{showcomments}
\setboolean{showcomments}{false}
\ifthenelse{\boolean{showcomments}}
{ \newcommand{\mynote}[3]{
    \fbox{\bfseries\sffamily\scriptsize#1}
    {\small$\blacktriangleright$\textsf{\textit{\color{#3}{#2}}}$\blacktriangleleft$}}}
{ \newcommand{\mynote}[3]{}}

\DontPrintSemicolon

\SetFuncSty{sc}
\SetDataSty{em}
\SetCommentSty{em}

\SetKwFor{func}{procedure}{}{end}
\SetKwFor{upon}{upon}{}{end}
\SetKwFor{struct}{struct}{}{end}
\SetKwFor{periodically}{task every $\delta$ time units}{}{end}
\SetKwFor{initially}{initially}{}{end}

\SetKw{Break}{break}
\SetKw{return}{return}

\SetKwFunction{BALL}{BALL}
\SetKwFunction{PULL}{PULL}
\SetKwFunction{broadcast}{Broadcast}
\SetKwFunction{receive}{receive}
\SetKwFunction{deliver}{deliver}
\SetKwFunction{Random}{Random}
\SetKwFunction{send}{send}
\SetKwFunction{dest}{to}
\SetKwFunction{orderEvents}{orderEvents}
\SetKwFunction{receiveEvents}{receiveEvents}
\SetKwFunction{isDeliverable}{isDeliverable}
\SetKwFunction{isPush}{\underline{isPush}}
\SetKwFunction{isPull}{\underline{isPull}}
\SetKwFunction{getClock}{getClock}
\SetKwFunction{updateClock}{updateClock}

\SetKwBlock{Block}{}{}

\title{Optimal epidemic dissemination}

\author{Hugues Mercier}
\email{hugues.mercier@unine.ch}
\author{Laurent Hayez}
\email{laurent.hayez@unine.ch}
\affiliation{%
  \institution{Université de Neuchâtel}
  \city{Neuchâtel}
  \country{Switzerland}
}
\author{Miguel Matos}
\affiliation{%
  \institution{INESC-ID \& IST, Universidade de Lisboa}
  \city{Lisboa}
  \country{Portugal}
}

\begin{abstract}
{
We consider the problem of reliable epidemic dissemination of a rumor in a fully connected network of~$n$ processes using push and pull operations. We revisit the random phone call model and show that it is possible to disseminate a rumor to all processes with high probability using $\Theta(\ln n)$ rounds of communication and only $n+o(n)$ messages of size $b$, all of which are asymptotically optimal and achievable with pull and push-then-pull algorithms. This contradicts two highly-cited lower bounds of Karp et al.~\cite{DBLP:conf/focs/KarpSSV00} stating that any algorithm in the random phone call model running in $\mathcal{O}(\ln n)$ rounds with communication peers chosen uniformly at random requires at least $\omega(n)$ messages to disseminate a rumor with high probability, and that any address-oblivious algorithm needs $\Omega(n \ln \ln n)$ messages regardless of the number of communication rounds. The reason for this contradiction is that in the original work~\cite{DBLP:conf/focs/KarpSSV00}, processes do not have to share the rumor once the communication is established. However, it is implicitly assumed that they always do so in the proofs of their lower bounds, which, it turns out, is not optimal. Our algorithms are strikingly simple, address-oblivious, and robust against $\epsilon n$ adversarial failures and stochastic failures occurring with probability $\delta$ for any $0 \leq \{\epsilon,\delta\} < 1$. Furthermore, they can handle multiple rumors of size $b \in \omega(\ln n \ln \ln n)$ with $nb + o(nb)$ bits of communication per rumor. 

}
\end{abstract}

\thanks{A brief announcement of this work was presented at PODC 2017.}

\begin{document}

\maketitle

\section{Introduction}
\label{sec:introduction}

We consider the problem of reliable epidemic/gossip dissemination of a rumor in a fully connected network of $n$ processes using address-oblivious algorithms. In this class of algorithms, the local decisions taken locally by each process are oblivious to the addresses of the other processes. Besides dissemination~\cite{bimodal}, epidemic/gossip-based algorithms have been proposed to address a wide variety of problems such as replicated database maintenance~\cite{Demers:1987:EAR:41840.41841}, failure detection~\cite{vanRenesse:1998:GFD:1659232.1659238},
aggregation~\cite{Kempe:2003:GCA:946243.946317}, code propagation and maintenance~\cite{Levis:2004:TSA:1251175.1251177}, modeling of computer virus propagation~\cite{Berger:2005:SVI:1070432.1070475}, membership~\cite{Jelasity2007},  publish-subscribe~\cite{10.1109/TPDS.2013.6}, total ordering~\cite{matos2015epto}, and numerous distributed signal processing tasks~\cite{journals/pieee/DimakisKMRS10}. The randomness inherent to the selection of the communication peers makes epidemic algorithms particularly robust to all kinds of failures such as message loss and process failures, which tend to be the norm rather than the exception in large systems. Their appeal also stems from their simplicity and highly distributed nature. 
The amount of work studying theoretical models of epidemic dissemination is vast and mainly focuses on establishing bounds for different dissemination models, which we briefly describe below.

\paragraph{\textbf{Push algorithms.}}

The simplest epidemic dissemination algorithms are push-based, where processes that know the rumor propagate
it to other processes.  Consider the following ``infect forever'' push algorithm first introduced by Frieze and Grimmett~\cite{DBLP:journals/dam/FriezeG85}. The algorithm starts with a
single process knowing a rumor, and at every round, every informed process chooses $\fanout$ processes uniformly at random and forwards the rumor to them. Pittel~\cite{Pittel:1987:SR:37387.37400} showed that for a network of size $n$,
$\log_{\fanout+1}n+\frac{1}{\fanout}\ln n + \mathcal{O}(1)$ rounds of communication are necessary and
sufficient in probability for every process to learn the rumor.
There are other flavors of push algorithms~\cite{EGMM2004,koldehofe2004simple}, although in all cases,
reaching the last few uninformed processes becomes increasingly costly as most messages are sent to processes already
informed.  Push algorithms must transmit $\Theta(n \ln n)$ messages if every process is to learn a rumor with high probability\footnote{With high probability (w.h.p)
  means with probability at least $1-\mathcal{O}\left(n^{-c}\right)$ for a constant $c>0$.}.

\paragraph{\textbf{Pull algorithms}}

Instead of pushing a rumor, a different strategy is for an uninformed process to ask another process chosen at
random to convey the rumor if it is already in its possession. Pulling rumors was first proposed and studied
by Demers et al.~\cite{Demers:1987:EAR:41840.41841}, and further studied by Karp et
al.~\cite{DBLP:conf/focs/KarpSSV00}. 
Pulling algorithms are advantageous when rumors are frequently created because pull requests will more often than not reach processes with new rumors to share. However, issuing pull requests in systems with little activity result in useless traffic. 

\paragraph{\textbf{Push-pull algorithms and the (polite) random phone call model}}

The idea to push and pull rumors simultaneously was first considered by Demers et al.~\cite{Demers:1987:EAR:41840.41841}, and further studied in the seminal work of Karp et al.~\cite{DBLP:conf/focs/KarpSSV00} who considered the following random phone call model. At each round, each process randomly chooses an interlocutor and calls it.  If, say, Alice calls Bob, Alice pushes the rumor to Bob if she has it, and pulls the rumor from Bob if he has it. Establishing communication (the phone call itself) is free, and only messages that include the rumor are counted. It is paramount to note that in the original work~\cite{DBLP:conf/focs/KarpSSV00}, processes do not have to share the rumor once the communication is established, although it is implicitly assumed that they always do in the analysis of their lower bounds. We thus define the \emph{polite random phone call} model as it is used in the analysis of \cite{DBLP:conf/focs/KarpSSV00}, i.e., assuming that processes always share the rumor.
We generalize this model, including the right not to share the rumor, in the next section.

Using the polite random phone call model, Karp et al.~\cite{DBLP:conf/focs/KarpSSV00} presented an algorithm that transmits a rumor to every process with high probability using $\mathcal{O}(\ln n)$ rounds of communication and $\mathcal{O}(n \ln \ln n)$ messages. The idea is that the number of informed processes increases exponentially at each round until approximately $\frac{n}{2}$ processes are informed due to the push operations, after which the number of uninformed processes shrinks quadratically at each round due to the pull operations. The authors also prove that any algorithm in the polite random phone call model running in $\mathcal{O}(\ln n)$ rounds with communication peers chosen uniformly at random requires at least $\omega(n)$ messages, and that any address-oblivious algorithm needs $\Omega(n \ln \ln n)$ messages to disseminate a rumor regardless of the number of communication rounds. Even though these lower bounds are valid in this polite random phone call model, the authors imply that they are valid in the more general model that they defined, which is false. We break both lower bounds in this article.

The work of Karp et al.~\cite{DBLP:conf/focs/KarpSSV00} is widely cited. Their push-pull algorithm is leveraged as a primitive block in numerous settings, but more worrisome, their lower bounds are wrongly used as fundamental limits of epidemic dissemination algorithms, which sometimes lead to cascaded errors. A relevant example here is the work of Fraigniaud and Giakkoupis~\cite{DBLP:conf/spaa/FraigniaudG10} on the total number of bits exchanged in the random phone call model. The authors presented a push-pull algorithm with concise feedback that requires $\mathcal{O}(\ln n)$ rounds and $\mathcal{O}(n(b + \ln \ln n \ln b))$ bits to disseminate a rumor of size $b$, as well as a lower bound of  $\Omega(nb + n \ln \ln n))$ bits when the number of rounds is in $\mathcal{O}(n)$. They proved the $nb$ term of the lower bound for $n \in \omega(\ln \ln n)$, but relied on the false $\Omega(n \ln \ln n)$ bound of ~\cite{DBLP:conf/focs/KarpSSV00} for the other term. Their correct lower bound is therefore $\Omega(nb)$, and only valid for $n \in \omega(\ln \ln n)$. 

\subsection{Our contributions}

\paragraph{\textbf{Generalized (impolite) random phone call model.}}

In the proofs of the original random phone call model, rumors are transmitted in both directions whenever both players on the line have the rumor. Our generalized model removes this restriction, and also allows multiple push and pull phone calls per round. Let $\fanout \geq 1$ and $\fanin \geq 1$. At each communication round, each process: i) calls between 0 and $\fanin$ processes uniformly at random to request a rumor, ii) calls between 0 and $\fanout$ processes uniformly at random to push a rumor, and iii) has the option not to answer pull requests. To keep the phone call analogy, our generalized model allows impolite parties: each player can call multiple players, refuse to reply to pull requests, refuse to push a rumor, and refuse to request a rumor at any given round. 

We assume, like for the original model, that establishing the communication is free, and we only count the number of messages that contain the rumor. The practical rationale behind this assumption is that the cost of establishing the communication is negligible if the rumor is large or if there are multiple rumors that can be transmitted in a single communication. 
We also assume that the network is a complete graph, 
that the rounds are synchronous, and that processes can reply to pull requests in the same round. Finally, we assume that a single process has a rumor to share at the start of the dissemination process\footnote{We handle multiple rumors over a long period of time in Section~\ref{sec:multiplerumors}.}. 

We define three regular algorithms, all defined to halt after an agreed upon number of dissemination rounds. 
In the \emph{regular pull algorithm} uninformed processes send exactly $\fanin$ pull requests per round, whereas informed processes never push, never send pull requests but always reply to pull requests.  In the \emph{regular push algorithm} informed processes push the rumor to exactly $\fanout$ processes per round, whereas uninformed processes never send pull requests. Finally, the \emph{regular push-then-pull algorithm} consists of a regular push algorithm followed by a regular pull algorithm. Note that the best protocols for the generalized random phone call model are strikingly simple and do not require, for instance, to define a complicated probability distribution that determines who replies to what: we prove that the regular pull algorithm and the regular push-then-pull algorithm are asymptotically optimal.

\paragraph{\textbf{Breaking the lower bounds from~\cite{DBLP:conf/focs/KarpSSV00}.}}

The confusion from the lower bounds of Karp et al.~\cite{DBLP:conf/focs/KarpSSV00} stems from the fact that their model definition allows impolite behavior, but the proofs of their lower bounds implicitly assume that processes always behave politely. More precisely, one the one hand, (1) they define the model such that processes do not have to share the rumor once the communication is established: ``Whenever a connection is established between two players, each one of them (if holding the rumor) has to decide whether to transmit the rumor to the other player, typically without knowing whether this player has received the rumor already.'' and (2) state their lower bounds as such: ``[...] any address-oblivious algorithm [...] needs to send $\Omega(n \ln \ln n)$ messages for each rumor regardless of the number of rounds. Furthermore, we give a general lower bound showing that time- and communication-optimality cannot be achieved simultaneously using random phone calls, that is, every algorithm that distributes a rumor
in $\mathcal{O}(\ln n)$ rounds needs $\omega(n)$ transmissions.'' On the other hand, in the proofs of their lower bounds in Theorems 4.1 and 5.1 it is implicitly assumed that processes always pull and push the rumor each time a communication is established. This is not optimal and allows us to break both lower bounds. The idea that selectively not replying and not pushing might be beneficial is never discussed.

\paragraph{\textbf{Optimal algorithms with  $\mathcal{O}(\ln n)$ rounds and $n+o(n)$ messages of size $b$}}

If we discount the cost of establishing the communication (the phone call), it is natural to let processes choose whether or not to call, and whether or not to reply when called. This generalization makes a huge difference: we show that the regular pull and push-then-pull algorithms disseminate a rumor of size $b$ to all processes with high probability in $\mathcal{O}(\ln n)$ rounds of communication using only $n+o(n)$ messages of size $b$. The idea is simple: we do not push old rumors because doing so results in a large communication overhead. 

Consider the regular pull algorithm. We prove that this algorithm requires $\Theta(\log_{\fanin+1} n)$ rounds of communication, $n-1$ messages of size $b$ when $\fanin=1$, and $\mathcal{O}(n)$ messages if $\fanin \in \mathcal{O}(1)$. This algorithm is optimal for the generalized phone call model. First, its message complexity is optimal since any algorithm requires at least $n-1$ messages. Second, its bit complexity is optimal for $b\in\omega(\ln \ln n)$ from the (corrected) $\Omega(nb)$ lower bound of Fraigniaud and Giakkoupis~\cite{DBLP:conf/spaa/FraigniaudG10}. 
Third, if $f=\fanin = \fanout$, we prove that its round complexity is asymptotically optimal by showing that pushing and pulling at the same time using potentially complex rules is unnecessary: any algorithm in the generalized random phone call model requires $\Omega(\log_{f+1} n)$ rounds of communication to disseminate a rumor with high probability.

Despite its utter simplicity, the regular pull algorithm exhibits strong robustness against adversarial and stochastic failures. Let $\delta$ be the probability that a phone call fails, and let $\epsilon \cdot n$ be a set of processes, excluding the process initiating the rumor, initially chosen by an adversary to fail at any point during the execution of the algorithm. We prove that for any $0 \leq \epsilon < 1$ and  $0 \leq \delta < 1$, $\mathcal{O}(\log_{\fanin+1} n)$ rounds of communication remain sufficient to inform all processes that do not fail with high probability. The number of transmitted messages when failures occur remains asymptotically optimal.

Although pushing is never required asymptotically, in practice the best approach is to push when the rumor is young until the expected communication overhead reaches an agreed upon threshold, and then pull until all processes learn the rumor with the desired probability. The regular  push-then-pull algorithm is thus asymptotically optimal when $\fanin \in \mathcal{O}(1)$ as long as the number of messages transmitted during the push phase is in $\mathcal{O}(n)$.

We also prove that when $b \in \omega(\ln n \ln \ln n)$, the regular pull and push-then-pull algorithms can be modified to handle multiple and possibly concurrent rumors over a long period of time with $nb + o(nb)$ bits of communication per rumor. This is optimal as it matches the $\Omega(nb)$ lower bound of~\cite{DBLP:conf/spaa/FraigniaudG10}.

The rest of this article is organized as follows. We present related work in Section~\ref{sec:epidemics}, followed by an analysis of pull algorithms in Section~\ref{sec:pull}.
We discuss push--pull algorithms in Section~\ref{sec:pushthenpull} and handle multiple rumors in Section~\ref{sec:multiplerumors}.


\section{Related work}
\label{sec:epidemics}

Multiple approaches have been proposed to overcome the overhead (number of messages, 
number of rounds and number of transmitted bits) of epidemic dissemination algorithms, especially the two lower bounds of Karp et al.~\cite{DBLP:conf/focs/KarpSSV00}. By allowing direct addressing in the random phone call model, Avin and Elsässer~\cite{avin2013faster} presented an algorithm requiring $\mathcal{O}(\sqrt{\ln n})$ rounds by building
a virtual topology between processes, at the cost of transmitting a larger number of more complex
messages. Haeupler and Malkhi~\cite{haeupler2014optimal} generalized the work with a gossip algorithm running
in $\mathcal{O}(\ln \ln n)$ rounds and sending $\mathcal{O}(1)$ messages per node with $\mathcal{O}(\ln n)$
bits per message, all of which are optimal. The main insight of their algorithm is the careful construction and
manipulation of clusters. Panagiotou et al.~\cite{panagiotou2013faster} removed the uniform
assumption of the random phone call model and presented a push-pull protocol using $\Theta(\ln \ln n)$
rounds. The number of calls per process is fixed for each process, but follows a power law distribution with
exponent $\beta \in (2, 3)$. This distribution has infinite variance and causes uneven load balancing, with
some processes that must call $\mathcal{O}(n)$ processes at every round. Doerr and
Fouz~\cite{doerr2011asymptotically} presented a push-only protocol spreading a rumor in $(1 + o(1))\log_2 n$
rounds and using $\mathcal{O}(nf(n))$ messages for an arbitrary function $f \in \omega(1)$. It assumes that each process possesses a permutation of all the processes. Doerr et al. \cite{doerr2016simple} disseminate information by randomizing the whispering protocols of \cite{Gasieniec96adaptivebroadcasting,Diks2000}.  Alistarh et al.~\cite{DBLP:conf/icalp/AlistarhGGZ10} designed a gossip protocol with a $\mathcal{O}(n)$ message complexity by randomly selecting a set of coordinators that collect and disseminate the rumors using overlay structures. Their algorithm is robust against oblivious failures. Processes are allowed to keep a communication line open over multiple rounds and can call $\mathcal{O}(n)$ processes per round.

Work on epidemic dissemination was done in other contexts and with different constraints, such as topologies other than the complete graph~\cite{fountoulakis2010rumor,DBLP:conf/stacs/Giakkoupis11}, communication with latency~\cite{DBLP:conf/podc/GilbertRS17} and asynchronicity~\cite{acan2015push}. 


\section{The regular pull algorithm is asymptotically optimal}
\label{sec:pull}

In this section, we focus on pull-only algorithms. 
Our first observation is that on expectation, pulling is always at least as good as pushing, although the higher variance of pull at the early stage of the dissemination makes pulling less efficient when the rumor is new. For instance, starting with one informed process and $\fanin=\fanout=1$, it takes $\Theta(\ln n)$ pull rounds to inform a second process with high probability, whereas a single push round suffices. The behavior reverses when the rumor is old: if $n-1$ processes are already informed, a single pull round informs the last process with high probability but $\Theta(\ln n)$ push rounds are needed. Despite these differences, our second observation is that pulling and pushing have the same asymptotic round complexity. Our third observation is that the regular pull algorithm is asymptotically optimal, thus pushing is not required. 
Our fourth observation is that the regular pull algorithm asymptotically requires the same round, bit, and message complexity even in the presence of a large number of adversarial and stochastic failures.

Note that in the generalized random phone call model, processes push and pull requests uniformly at random but independently (i.e., with replacement), thus they can push the rumor to themselves, call themselves, and have multiple push messages and/or pull requests colliding in the same round. Of course in practice, in a given round, a process will not send multiple pull requests or multiple push messages to the same process, nor will it call itself. Instead, it will select a uniform random sample among the other processes in the network. Our reason for this definition is twofold. First, choosing interlocutors independently and uniformly at random is more amenable to mathematical analysis, especially upper bounds. Second, we prove that choosing $f$ processes uniformly at random with replacement, or choosing a uniform random sample of size $f$ without replacement among the other $n-1$ processes, are asymptotically equivalent when $f \in \mathcal{O}(n)$. We prove this by matching lower bounds obtained from random samples with upper bounds obtained with interlocutors selected independently and uniformly at random.
\begin{definition}
  Let $0 \leq \uninformed{r} \leq n$ be the number of uninformed processes at round $r$,  $\E_{\text{pull}}[\uninformed{r}]$ the expected number of uninformed processes at round $r$ with the regular pull algorithm, and $\E_{\text{push}}[\uninformed{r}]$ the expected number of uninformed processes at round $r$ with the regular push algorithm. For the number of informed processes at round $r$, we similarly define $\informed{r}$, $\E_{\text{pull}}[\informed{r}]$ and $\E_{\text{push}}[\informed{r}]$. It is clear that $n=\uninformed{r}+\informed{r}=\E_{\text{pull}}[\uninformed{r}]+\E_{\text{pull}}[\informed{r}]=\E_{\text{push}}[\uninformed{r}]+\E_{\text{push}}[\informed{r}]$.
\end{definition}
If processes send pull requests independently and uniformly at random, $\P(\uninformed{r+1} \mid \uninformed{r})$ follows a binomial distribution with mean
\begin{align}
  \label{eq:meanpull1}
  \E_\text{pull}[\uninformed{r+1} \mid \uninformed{r}]&=\uninformed{r}\cdot\left(\frac{\uninformed{r}}{n}\right)^{\fanin}
\end{align}
whereas if they select a uniform random sample without replacement among the other $(n-1)$ processes we obtain
\begin{align}
  \label{eq:meanpull2}
  \E_\text{pull}[\uninformed{r+1} \mid \uninformed{r}]&=\uninformed{r}\cdot\frac{{\uninformed{r}\choose\fanin}}{{n-1 \choose \fanin}}=n-\uninformed{r}\frac{\uninformed{r}(\uninformed{r}-1)\dots(\uninformed{r}-\fanin+1)}{n(n-1)\dots(n-\fanin+1)} = n-\uninformed{r} \frac{(\uninformed{r})_{\fanin}}{(n-1)_{\fanin}} 
\end{align}
where $(\boldsymbol{\cdot})_{\boldsymbol{\cdot}} $ is the falling factorial notation.

\begin{lemma}
  \label{lem:pullbetterthanpush}
If $\fanout = \fanin$, then 
    $\E_{\text{pull}}[\uninformed{r+1}|\uninformed{r}] \leq \E_{\text{push}}[\uninformed{r+1}|\uninformed{r}]$.
\end{lemma}

\begin{proof}
We prove the lemma with processes chosen independently and uniformly at random.  Let  $f=\fanin=\fanout$. For the pull version, we saw that
\begin{align}
  \label{eq:meanpull1repeat}
  \E_\text{pull}[\uninformed{r+1} \mid \uninformed{r}]=\uninformed{r}\cdot\left(\frac{\uninformed{r}}{n}\right)^{f}
\end{align}
whereas for the push version we can show that 
  \begin{equation}
    \label{eq:meanpush}
   \E_{\text{push}}[\uninformed{r+1} \mid \uninformed{r}] = \uninformed{r}\left(1 - \frac{1}{n}\right)^{f(n-\uninformed{r})}.
 \end{equation}
 From Eq.~\eqref{eq:meanpull1repeat} and \eqref{eq:meanpush}, it is clear that the lemma holds when $\uninformed{r}=0$, $\uninformed{r}=n-1$, and $\uninformed{r}=n$. For the other values of $\uninformed{r}$, we prove that
 \begin{align}
\label{eq:comp}
      & \left( \frac{\uninformed{r}}{n}\right)^{f} \leq \left( \left(1 - \frac{1}{n}\right)^{n - \uninformed{r}}\right)^{f} 
      \Leftrightarrow \left(\frac{n-1}{n}\right)^{n - \uninformed{r}} - \frac{\uninformed{r}}{n} \geq 0.
 \end{align}
  Let $g(x) \triangleq \left(\frac{n-1}{n}\right)^{n - x} - \frac{x}{n}$. Since $g(0) \geq 0$ and $g(n-1) = 0$, we prove that $g(x) \geq 0$ for every $x \in \{0,1,\dots, n-1\}$ by showing that $g'(x) \leq 0$ over the interval $[0,n-1]$. We have
  \begin{equation}
     \begin{split}
      g'(x) 
        &= - \left( \frac{n-1}{n} \right)^{n-x} \ln \left(\frac{n-1}{n}\right) - \frac{1}{n} \\
        &= \left(\frac{n}{n-1}\right)^x \left( \frac{n-1}{n} \right)^{n} \ln \left(\frac{n}{n-1}\right) - \frac{1}{n} \\
      \end{split}
    \end{equation}
    which is an increasing function with respect to $x$. To complete the proof, we verify that $g'(n-1) \leq 0$:
     \begin{equation}
    \begin{split}
      g'(n-1) & =  \left(\frac{n}{n-1}\right)^{(n-1)} \left( \frac{n-1}{n} \right)^{n} \ln \left(\frac{n}{n-1}\right) - \frac{1}{n} \\
      & \leq \frac{n-1}{n} \left(\frac{n}{n-1} -1\right)  - \frac{1}{n} \\ & = 0.
    \end{split}
  \end{equation}
\end{proof}

 We now bound the expected progression of the regular pull algorithm, and later use it to derive lower bounds on its round complexity.
 
 \begin{lemma}
   \label{lem:pullasymptotic}
$\E_\text{pull}[\informed{r+1}\mid\informed{r}] \leq \informed{r} \cdot(\fanin+1)$.
 \end{lemma}

\begin{proof}
  We prove the lemma with processes chosen from a uniform random sample using Eq.~\eqref{eq:meanpull2}. We fix $n$ and $\uninformed{r}$ and prove the lemma by induction on $\fanin$.

  \paragraph{Basis step.} The lemma is clearly true for $x=\fanin=0$.
  \paragraph{Inductive step.} Let $0 \leq x \leq n-2$ be an integer. We assume that $n-\uninformed{r} \frac{(\uninformed{r})_x}{(n-1)_x} \leq \informed{r}(x+1)$, which is equivalent to
  \begin{align}
    \label{eq:ind1}
     \uninformed{r} \frac{(\uninformed{r})_x}{(n-1)_x} \geq n - \informed{r}(x+1)
  \end{align}
  and must show that
  \begin{align}
    \label{eq:ind2}
    n-\uninformed{r} \frac{(\uninformed{r})_x(\uninformed{r}-x)}{(n-1)_x(n-1-x)} \leq \informed{r}(x+2) \Leftrightarrow  S \triangleq n-\uninformed{r} \frac{(\uninformed{r})_x(\uninformed{r}-x)}{(n-1)_x(n-1-x)} - \informed{r}(x+2) \leq 0.
     \end{align}
Substituting the left side of Eq.~(\ref{eq:ind1}) for its right side in Eq.~(\ref{eq:ind2}), and replacing $\uninformed{r}$ by $n - \informed{r}$, we have

\begin{equation}
  \begin{split}
  S & \leq n- (n - \informed{r}(x+1)) \frac{n-\informed{r}-x}{n-1-x} - \informed{r}(x+2) \\
    & \leq \frac{n(\informed{r}-1) - \informed{r} (x+1) (\informed{r}  - 1)}{n-x-1} - \informed{r} \\ & \leq \informed{r}-1 - \frac{(x+1) (\informed{r}  - 1)^2}{n-x-1} - \informed{r} \\
    & \leq - \frac{(x+1) (\informed{r}  - 1)^2}{n-x-1} \\
  & \leq 0.
  \end{split}
\end{equation}
  
 \end{proof}

\begin{lemma}
    \label{lem:loglog}
        If $\fanin\in\mathcal{O}(\ln n)$, the regular pull algorithm starting with $\frac{n}{\ln n}$ informed processes informs all processes with high probability in $\Theta(\log_{\fanin+1} \ln n)$ rounds.
  \end{lemma}

  \begin{proof}
For the lower bound, it is clear from Lemma~\ref{lem:pullasymptotic} that $\Omega(\log_{\fanin+1} \ln n)$ are required to reach all processes on expectation, thus required to inform all processes with high probability. For the upper bound, the proof for $\fanin=1$ consists of the points 3 and 4 in the proof of Theorem 2.1 of Karp et
    al.~\cite{DBLP:conf/focs/KarpSSV00}. We generalize their proof for an arbitrary $\fanin$.

    Recall that $\E_{\text{pull}}[\uninformed{t}\mid\uninformed{t-1}] =  \frac{(\uninformed{t-1})^{\fanin+1}}{n^{\fanin}}$ and that we start with at most $u_0 = n - \frac{n}{\ln n}$ uninformed processes. We use the following Chernoff bound from \cite{mitzenmacher2005probability}: 
    \[ \P(X \geq (1 + \delta) \mu) \leq e^{-\frac{\delta^2 \mu}{3}},\ 0 < \delta < 1. \]
If $u_{t-1} \geq (\ln n)^{\frac{4}{\fanin+1}}n^\frac{\fanin}{\fanin+1}$, it follows that
    \begin{align*}
      \P\left(u_t \geq \left(1 + \frac{1}{\ln n}\right) \frac{(\uninformed{t-1})^{\fanin+1}}{n^{\fanin}}\right) 
        & \leq  e^{- \frac{1}{3} \ln^2 n} \\ & \in o\left(n^{-c}\right) \text{ for any constant $c$}
    \end{align*}
and we can deduce that
    \begin{equation}
      \label{eq:uninformed-whp}
      \uninformed{t} \leq \left(1 + \frac{1}{\ln n}\right) \frac{(\uninformed{t-1})^{\fanin+1}}{n^{\fanin}}
    \end{equation}
    with high probability. Applying Eq. (\ref{eq:uninformed-whp}) recursively, we obtain
    \begin{align}
      \uninformed{t} & \leq (\uninformed{0})^{{(\fanin+1)^t}}\left(\frac{1 + \frac{1}{\ln n}}{n^{\fanin}}\right)^\frac{(\fanin+1)^t-1}{\fanin} 
    \end{align}
Replacing $u_0$ by $n - \frac{n }{\ln n}$, and $t$ by $4 \log_{\fanin+1}\ln n$ we obtain 
\begin{equation}
  \begin{split}
   \uninformed{t} & \leq \left(n-\frac{n}{\ln n}\right)^{{(\fanin+1)^t}}\left(\frac{1 + \frac{1}{\ln n}}{n^{\fanin}}\right)^\frac{(\fanin+1)^t-1}{\fanin} \\ & \leq n \left(1-\frac{1}{\ln n}\right)^{\ln^4 n} \left(1+\frac{1}{\ln n}\right)^{\ln^4 n} \\ & \leq n \left(1-\frac{1}{\ln^2 n}\right)^{\ln^4 n} \\ & \in o(1)
 \end{split}
\end{equation}
which shows that we need $O(\log_{\fanin+1} \ln n)$ rounds to reach the point where there are at most $(\ln n)^{\frac{4}{\fanin+1}}n^\frac{\fanin}{\fanin+1}$ uninformed processes with high probability. Note that this step is unnecessary if $\fanin$ is large enough with respect to $n$ since $(\ln n)^{\frac{4}{\fanin+1}}n^\frac{\fanin}{\fanin+1} \geq n-\frac{n}{\ln n}$.

At this stage, the probability that an uninformed process remains uninformed after each subsequent round is at most 
    \begin{align}
      \label{eq:21}
      \left(\frac{\uninformed{r}}{n}\right)^{\fanin} & \leq \left(\frac{(\ln n)^{\frac{4}{\fanin+1}}n^\frac{\fanin}{\fanin+1}}{n}\right)^{\fanin} \leq \frac{(\ln n)^4}{\sqrt{n}}.
    \end{align}
    Hence after a constant number of additional rounds, we inform every remaining uninformed process with high probability. 

  \end{proof}


  \begin{corollary}
 If $\fanin\in\Omega(\ln n)$, the regular pull algorithm starting with $\frac{n}{\ln n}$ informed processes informs all processes with high probability in $\Theta(1)$ rounds.
  \end{corollary}


\begin{theorem}
  \label{thm:pull}
    The regular pull algorithm disseminates a rumor to all processes with high probability in $\Theta(\log_{\fanin + 1} n)$ rounds of communication.
\end{theorem}

\begin{proof}
  For the lower bound, it is clear from Lemma~\ref{lem:pullasymptotic} that $\Omega(\log_{\fanin+1} n)$ rounds are required in expectation to inform all processes, and thus necessary to inform all processes with high probability.

  We now show that $\mathcal{O}(\log_{\fanin + 1} n)$ rounds suffice when
  $\fanin \in \mathcal{O}(\ln n)$ (the statement for $\fanin = 1$ is implicitly discussed without proof in \cite{DBLP:conf/focs/KarpSSV00}). 

  In a first phase, we show that $\mathcal{O}(\log_{\fanin+1}n)$ rounds are sufficient to inform $\ln n$ processes with high probability. Let $c_0 \geq 1$ be a constant. In this case, we show that for stages $k \in \{0,1,2,\dots,\ln\ln n\}$, if $\informed{r} = 2^k$ processes are informed, then after $\rho_k \triangleq c_0\left\lceil \frac{\log_{\fanin + 1} n}{2^k} \right \rceil$ rounds, the number of informed processes doubles with high probability, i.e., $\informed{r+\rho_k} \geq 2^{k+1}$ with high probability. At every round of stage $k$, each pull request has a probability at least $\frac{2^k}{n}$ of reaching an informed process, thus after $\rho_k$ rounds and $\rho_k \cdot \fanin$ pull requests, the probability that an uninformed process learns the rumor is bounded by
  \begin{align}
  p \geq 1 - \left( 1 - \frac{2^k}{n} \right)^{\rho_k \cdot \fanin} \geq \frac{2^k\rho_k\fanin}{n}-\frac{2^{2k}\rho_k^2\fanin^2}{n^2}.    
  \end{align}
  The probability $T$ to inform $l=\informed{r}=2^k$ processes or less in stage $k$ is upper bounded by the left tail of the binomial distribution with parameters $p$ and $N = \uninformed{r} = n-2^k$. We can bound this tail using the Chernoff bound
  \begin{align}
    \label{eq:chernoffbinomial}
       T & \leq \exp\left(-\frac{(Np-l)^2}{2Np}\right)
  \end{align}
  which is valid when $l \leq N p$. We can indeed apply this bound by showing that $Np \geq \frac{c_0 \fanin}{\ln(\fanin + 1)} \ln n + o(1)$, which is greater than $2^k$ when $c_0 \geq 1$. The Chernoff bound gives
  \begin{equation}
       \begin{split}
     \label{eq:3}
      T 
     & \leq \exp\left(-\frac{Np}{2}+l\right) \\
        & \leq \exp \left(-\frac{(n-2^k)p}{2} + 2^{k}\right) \\     
    & \leq \exp\left(\left(1-\frac{c_0\fanin}{2\ln(\fanin+1)}\right) \ln n +o(1)\right) \\
  & \in \mathcal{O}\left(n^{  1-\frac{c_0\fanin}{2\ln(\fanin+1)}}  \right)
\end{split}
\end{equation}
and for any constant $c>0$ we can find $c_0$ such that $T \in \mathcal{O}\left(n^{-c}\right)$. This first phase, with the $k$ stages, requires $\sum\limits_{k=0}^{\ln\ln n} \rho_k \leq c_0 \log_{\fanin + 1} n \cdot \sum\limits_{k=0}^{\ln\ln n} 2^{-k} + c_0(\ln\ln n + 1) \sim 2c_0 \log_{\fanin + 1} n$ rounds of communication to inform $1 + 2^0 + 2^1 + \dots + 2^{\ln\ln n} \approx 2 \ln n$ processes with high probability. 

In a second phase, when $\ln n \leq \informed{r} \leq \frac{n}{(\ln n)^2}$, we show that a constant number of rounds $c_1$ is sufficient to multiply the number of informed processes by $\fanin+1$ with high probability. We use the Chernoff bound of Eq.~\eqref{eq:chernoffbinomial} with $l=\fanin \cdot\informed{r}$, $n-\frac{n}{\ln n}\leq N \leq n-\ln n$ and $p \geq 1-\left(1-\frac{\informed{r}}{n}\right)^{c_1 \fanin} \geq \frac{\informed{r}c_1\fanin}{n} -\frac{\informed{r}^2 c_1^2 \fanin^2}{2n^2}.$ We obtain
\begin{equation}
  \begin{split}
  T & \leq \exp\left(-\frac{Np}{2}+l\right) \\
    & \leq \exp \left( -\frac{\informed{r}c_1\fanin}{2}\left(1-o(1)\right)+\informed{r}\fanin\right) \\
    & \leq \exp \left(\ln n\left(1-\frac{c_1}{2}+o(1)\right) \right) \\
    & \in \mathcal{O}\left(n^{  1-\frac{c_1}{2}}  \right)
  \end{split}
\end{equation}
and for any constant $c>0$ we can find $c_1$ such that $T \in \mathcal{O}\left(n^{-c}\right)$.  This second phase requires $\mathcal{O}(\log_{\fanin + 1} n)$ rounds of communication.

In a third phase, we can go from $\frac{n}{(\ln n)^2}$ to $\frac{n}{\ln n}$ informed processes in $\mathcal{O}(\log_{\fanin + 1} n)$ rounds of communication since multiplying the number of informed processes by $\ln n$ at this stage cannot be slower than during the first phase.  Finally, in a fourth phase we saw in Lemma~\ref{lem:loglog} that we can go from $\frac{n}{\ln n}$ to $n$ informed processes with high probability with $\Theta\left(\log_{\fanin+1}\ln n\right)$ rounds of communication.

We now summarize the proof of the upper bound when $\fanin\in \omega(\ln n)$ and $\fanin\in\mathcal{O}(n)$. The different cases must me handled with care, but we omit the details for simplicity purposes. In a first phase, we show that  $\mathcal{O}(\log_{\fanin + 1} n)$  rounds are sufficient to inform $\ln n$ processes with high probability. In a second phase, if $\fanin \cdot \informed{r} \in o(n)$, we apply the Chernoff bound of Eq.~\eqref{eq:chernoffbinomial} during $\mathcal{O}(\log_{\fanin + 1} n)$  rounds to reach either $\frac{n}{\ln n}$ informed process with high probability, or $\fanin \cdot \informed{r} \in \Theta(n)$ (the Chernoff bound must be changed when $\fanin \cdot \informed{r} \in \Theta(n)$). If $\fanin \cdot i \in \Theta(n)$, we again apply Eq.~\eqref{eq:chernoffbinomial} during a constant number of rounds to reach $c_2 \cdot n$ informed processes with $c_2<1$ with high probability. Finally, in a last phase, we go from $c_2 \cdot  n$ or $\frac{n}{\ln n}$ to $n$ informed processes with high probability using Lemma~\ref{lem:loglog}.

\end{proof}


\begin{corollary}
  If $\fanin \in \mathcal{O}(1)$, then the total number of messages (replies to pull requests) required by the regular pull algorithm is in $\Theta(n)$. In particular, the communication overhead is 0 when $\fanin=1$.
\end{corollary}

\begin{proof}
  It is clear that a process cannot pull a rumor more than $\fanin$ times since it stops requesting it in the rounds that follow its reception.
\end{proof}


We now prove that the round complexity of the regular pull algorithm is asymptotically optimal for the generalized random phone call model.

\begin{theorem}
  \label{thm:pushpullround}
If $f=\fanin=\fanout$, any protocol in the generalized random phone call model requires $\Omega(\log_{f+1} n)$ rounds of communication to disseminate a rumor to all processes with high probability.
\end{theorem}

\begin{proof}
Let $f=\max(\fanin,\fanout)$. 
If we only push messages, it is clear that the number of informed processes increases at most by a factor of $(\fanout + 1)$ per round. If we only pull messages, we saw in Lemma~\ref{lem:pullasymptotic}
  that the number of informed processes increases at most by a factor of $(\fanin + 1)$ per round in
  expectation. If all processes simultaneously push and pull at every round, the number of informed processes
  increases at most by a factor of $(\fanin + 1)(\fanout + 1)$ per round in expectation, thus the number or
  rounds required to informed all processes is at least
  $\log_{(\fanout+1)(\fanin+1)} n \geq \log_{(f+1)^2} n \in \Omega(\log_{f+1} n)$.

\end{proof}


We now show that the regular pull algorithm is robust against adversarial and stochastic failures.  First, consider an adversary that fails $\epsilon \cdot n$ processes for $0 \leq \epsilon < 1$, excluding the process starting the rumor. Before the execution of the algorithm, the adversary decides which processes fail, and for each failed process during which round it fails. Once a process fails, it stops participating until the end of the execution, although it may still be uselessly called by active processes. We also consider stochastic failures, in the sense that each phone call fails with probability $\delta$ for $0 \leq \delta < 1$. Note that both types of failures are independent of the execution.

The main difference introduced by the failures is that we can no longer go from $\frac{n}{\ln n}$ to $n$ informed processes in $\mathcal{O}(\log_{f+1} \ln n)$ rounds because there is a non-vanishing probability that pull requests either target failed processes or result in failed phone calls. We nevertheless show that the regular pull algorithm can disseminate a rumor to all $(1-\epsilon)n$ good (i.e., non-failed) processes with high probability with the same asymptotic round complexity. 

\begin{theorem}
  \label{thm:failures}
  Let $0 \leq \epsilon < 1$, and let $0 \leq \delta < 1$. If $\epsilon \cdot n$ processes, excluding the initial process with the rumor, fail adversarially, and if phone calls fail with probability $\delta$, then the regular pull algorithm still disseminates a rumor to all $(1-\epsilon)n$ good processes with high probability in $\Theta(\log_{\fanin + 1} n)$ rounds of communication.
\end{theorem}

\begin{proof}
  It is clear that the lower bound remains valid when there are failures. We prove the upper bound for $\fanin\in\mathcal{O}(\ln n)$, but as we mentioned for Theorem~\ref{thm:pull} we can adapt the proof for $\fanin\in\omega(\ln n)$ by carefully applying Chernoff bounds in different phases.

Note that the earlier a process fails, the more damage it causes. We thus assume that the $\epsilon \cdot n$ processes fail at the beginning of the execution, which is the worst possible scenario. We can use the first three phases of the proof of Theorem \ref{thm:pull} with minor modifications (only multiplicative constants change) and prove that $\mathcal{O}(\log_{\fanin + 1} n)$ rounds are sufficient to go from 1 to $\frac{n}{\ln n}$ informed processes with high probability. 

We now show that we need $c_2\log_{\fanin + 1} n$ rounds to go from $\frac{n}{\ln n}$ to $c_1 \cdot n$ informed processes with high probability for some arbitrary $c_1 < 1-\epsilon$. We again use the Chernoff bound of Eq.~\eqref{eq:chernoffbinomial} with $\informed{r}=\frac{n}{\ln n}$, $l = c_1 \cdot n$ and $N = (1-\epsilon) n - \frac{n}{\ln n}$. If $c_2$ is a large enough constant, the probability that a process learns a rumor during that phase is
  \begin{align}
  p \geq 1-\left(1-\frac{(1-\delta)\informed{r}}{n}\right)^{\fanin c_2\log_{\fanin + 1} n} \geq 1-\left(1-\frac{1-\delta}{\ln n}\right)^{\frac{\fanin c_2\ln n}{\ln(\fanin+1)}} \geq 1-e^{-c_2(1-\delta)} \triangleq c_3.
  \end{align}
The Chernoff bound gives
\begin{equation}
 \begin{split}
    T & \leq \exp\left(-\frac{Np}{2}+l\right) \\
      & \leq \exp \left( -c_3n\left(1-\epsilon-\frac{1}{\ln n}\right)+c_1 n \right)\\
     & \leq \exp \left( n\left(-c_3+c_3\epsilon+c_1+o(1) \right)\right)\\
   \end{split}
 \end{equation}
and we can choose $c_2$ such that $T \leq e^{c_4 n}$ with $c_4 < 0$. This guarantees $T \in \mathcal{O}\left(n^{-c}\right)$ for any~$c>0$.

Starting from $c_1 n$ informed processes, the probability that a process is informed in any subsequent round is bounded by $p \geq 1-\left(\frac{n-c_1(1-\delta)n}{n}\right)^{\fanin} \geq 1 - (1-c_1(1-\delta))^{\fanin}$.  After $r$ such rounds, the probability that a process remains uninformed is thus upper bounded by $(1-c_1(1-\delta))^{\fanin r}$, and for this probability to be bounded by $n^{-c}$ we need
  \begin{align}
    (1-c_1(1-\delta))^{\fanin r} \leq n^{-c} \Leftrightarrow
    r \geq  \frac{c \ln n}{\ln{\frac{1}{1-c_1(1-\delta)}}\fanin} \geq c_4 \log_{f+1} n \text{ for some constant $c_4$.}
  \end{align}
Hence, $\mathcal{O}(\log_{f+1} n)$ rounds are sufficient to go from $c_1n$ to $(1-\epsilon) n$ informed processes with high probability.

\end{proof}

Note that adversarial and stochastic failures do not increase the message complexity of the regular pull algorithm: uninformed processes that fail decrease the number of rumor transmissions, and failed phone calls do not exchange the rumor. We could, however, consider that messages containing the rumor are dropped with probability $0 \leq \gamma < 1$. Theorem~\ref{thm:failures} also holds in this instance, but the number of messages increases by an unavoidable factor of $\frac{1}{1-\gamma}$.





\section{The regular push-then-pull algorithm is asymptotically optimal}
\label{sec:pushthenpull}

As we demonstrated in this work, pushing is asymptotically unnecessary. Of course, practitioners have known for thirty years that it is preferable to push when the rumor is young, and to pull when the rumor is old~\cite{Demers:1987:EAR:41840.41841}.  It appears, however, that most researchers are unaware that pushing and pulling at the same time is not optimal. It also seems that both practitioners and researchers are unaware of the benefits of switching from the push to the pull phase \emph{early enough}.

The regular push-then-pull algorithm leverages the push and pull strategies when they are at their best, and decreases the prohibitive communication overhead caused by pushing messages to processes already informed. When $\fanin\in\mathcal{O}(1)$, as long as the communication overhead of the push phase is in $\mathcal{O}(n)$, the algorithm is asymptotically optimal. Note that even if pulling when the rumor is young incurs little overhead, one should substitute pull requests with push messages; for instance, a young rumor is more likely to propagate early using a regular push algorithm with $\fanout=2$ than if pushing and pulling at the same time with ${\fanout}=\fanin=1$.

Another advantage of the regular push-then-pull algorithm is that we can estimate with great precision the number of push rounds to reach a predefined communication overhead threshold. This is further discussed in the extended version of~\cite{DBLP:conf/podc/MercierHM17}, where it is proved, for instance, that running the push phase for $\log_{\fanout+1} n - \log_{\fanout+1} \ln n$ rounds guarantees that the communication overhead is in  $\mathcal{O}\left(\frac{n}{(\ln n)^2}\right)$. This makes no difference asymptotically compared to the regular pull algorithm, but in practice it bypasses the slow pull dissemination of young rumors while ensuring a bounded communication overhead. The number of messages quickly grows to $\omega(n)$ if the push phase is too long: with $\fanout=\fanin=1$, if we run the push phase during $\log_2 n +\Theta(\ln \ln n)$ rounds followed by a pull phase of $\mathcal{O}(\ln \ln n)$ rounds, the resulting push-then-pull algorithm exhibits the behavior of the seminal push-pull algorithms of Karp et al.~\cite{DBLP:conf/focs/KarpSSV00} and requires $\Theta(n \ln \ln n)$ messages. 

\section{Handling multiple rumors}
\label{sec:multiplerumors}

We can easily modify the regular pull and push-then-pull algorithms to handle multiple rumors of size $b$ over a long period of time as follows. First, processes append the age of the rumors to the messages containing them so that they know when to switch from the push to the pull phase, and when to stop their dissemination. If needed, these messages can also include the identifier of the process that first created the rumor to distinguish identical rumors initiated by multiple processes concurrently. Second, processes transmit the identity of the active rumors they already know with the pull requests to avoid receiving them multiple times during their pull phase. If $\fanin=1$, and if the overhead at the end of the push phase is in $o(n)$, then the resulting algorithms require the transmission of $n+o(n)$ messages containing each rumor and $\mathcal{O}(n \ln n \ln b) + (n+o(n)) (b + \ln \ln n)$ bits of communication per rumor. If $b$ is between $\omega(\ln \ln n)$ and $o(\ln n \ln \ln n)$, the push-pull algorithm with concise feedback of Fraigniaud and Giakkoupis~\cite{DBLP:conf/spaa/FraigniaudG10} using $\mathcal{O}(n (b + \ln \ln n \ln b ) ))$ bits is asymptotically better. However, if $b \in \omega(\ln n \ln \ln n)$, which is the case for most applications of interest, our algorithm requires $nb + o(nb)$ bits. This is optimal and better than the algorithm of~\cite{DBLP:conf/spaa/FraigniaudG10} which requires $c \cdot nb$ bits for a constant $c \geq 1$ based on the probability of imperfect dissemination. Again, asymptotically both solutions are equivalent, but we expect the simplicity of our approach and its multiplicative constant of 1 to make a significant difference for practical applications.


\bibliographystyle{unsrtnat}
\bibliography{bib.bib}

\begin{thebibliography}{30}
\providecommand{\natexlab}[1]{#1}
\providecommand{\url}[1]{\texttt{#1}}
\expandafter\ifx\csname urlstyle\endcsname\relax
  \providecommand{\doi}[1]{doi: #1}\else
  \providecommand{\doi}{doi: \begingroup \urlstyle{rm}\Url}\fi

\bibitem[Karp et~al.(2000)Karp, Schindelhauer, Shenker, and
  V{\"{o}}cking]{DBLP:conf/focs/KarpSSV00}
Richard~M. Karp, Christian Schindelhauer, Scott Shenker, and Berthold
  V{\"{o}}cking.
\newblock Randomized rumor spreading.
\newblock In \emph{41st Annual Symposium on Foundations of Computer Science,
  {FOCS} 2000}, pages 565--574, 2000.

\bibitem[Birman et~al.(1999)Birman, Hayden, {\"{O}}zkasap, Xiao, Budiu, and
  Minsky]{bimodal}
Kenneth~P. Birman, Mark Hayden, {\"{O}}znur {\"{O}}zkasap, Zhen Xiao, Mihai
  Budiu, and Yaron Minsky.
\newblock Bimodal multicast.
\newblock \emph{{ACM} Trans. Comput. Syst.}, 17\penalty0 (2):\penalty0 41--88,
  1999.

\bibitem[Demers et~al.(1987)Demers, Greene, Hauser, Irish, Larson, Shenker,
  Sturgis, Swinehart, and Terry]{Demers:1987:EAR:41840.41841}
Alan Demers, Dan Greene, Carl Hauser, Wes Irish, John Larson, Scott Shenker,
  Howard Sturgis, Dan Swinehart, and Doug Terry.
\newblock Epidemic algorithms for replicated database maintenance.
\newblock In \emph{Proceedings of the Sixth Annual ACM Symposium on Principles
  of Distributed Computing}, PODC '87, pages 1--12, 1987.

\bibitem[van Renesse et~al.(1998)van Renesse, Minsky, and
  Hayden]{vanRenesse:1998:GFD:1659232.1659238}
Robbert van Renesse, Yaron Minsky, and Mark Hayden.
\newblock A gossip-style failure detection service.
\newblock In \emph{Proceedings of the IFIP International Conference on
  Distributed Systems Platforms and Open Distributed Processing}, Middleware
  '98, pages 55--70, 1998.

\bibitem[Kempe et~al.(2003)Kempe, Dobra, and
  Gehrke]{Kempe:2003:GCA:946243.946317}
David Kempe, Alin Dobra, and Johannes Gehrke.
\newblock Gossip-based computation of aggregate information.
\newblock In \emph{44th Symposium on Foundations of Computer Science, {FOCS},
  Proceedings}, pages 482--491, 2003.

\bibitem[Levis et~al.(2004)Levis, Patel, Culler, and
  Shenker]{Levis:2004:TSA:1251175.1251177}
Philip Levis, Neil Patel, David~E. Culler, and Scott Shenker.
\newblock Trickle: {A} self-regulating algorithm for code propagation and
  maintenance in wireless sensor networks.
\newblock In \emph{1st Symposium on Networked Systems Design and Implementation
  {NSDI}}, pages 15--28, 2004.

\bibitem[Berger et~al.(2005)Berger, Borgs, Chayes, and
  Saberi]{Berger:2005:SVI:1070432.1070475}
Noam Berger, Christian Borgs, Jennifer~T. Chayes, and Amin Saberi.
\newblock On the spread of viruses on the internet.
\newblock In \emph{Proceedings of the Sixteenth Annual ACM-SIAM Symposium on
  Discrete Algorithms}, SODA '05, pages 301--310, 2005.

\bibitem[Jelasity et~al.(2007)Jelasity, Voulgaris, Guerraoui, Kermarrec, and
  van Steen]{Jelasity2007}
M{\'{a}}rk Jelasity, Spyros Voulgaris, Rachid Guerraoui, Anne{-}Marie
  Kermarrec, and Maarten van Steen.
\newblock Gossip-based peer sampling.
\newblock \emph{{ACM} Trans. Comput. Syst.}, 25\penalty0 (3):\penalty0 8, 2007.

\bibitem[Matos et~al.(2013)Matos, Felber, Oliveira, Pereira, and
  Riviere]{10.1109/TPDS.2013.6}
Miguel Matos, Pascal Felber, Rui Oliveira, José Pereira, and Etienne Riviere.
\newblock Scaling up publish/subscribe overlays using interest correlation for
  link sharing.
\newblock \emph{IEEE Transactions on Parallel and Distributed Systems},
  24\penalty0 (12):\penalty0 2462--2471, 2013.

\bibitem[Matos et~al.(2015)Matos, Mercier, Felber, Oliveira, and
  Pereira]{matos2015epto}
Miguel Matos, Hugues Mercier, Pascal Felber, Rui Oliveira, and Jos{\'e}
  Pereira.
\newblock {EpTO}: An epidemic total order algorithm for large-scale distributed
  systems.
\newblock In \emph{Proceedings of the 16th Annual Middleware Conference}, pages
  100--111. ACM, 2015.

\bibitem[Dimakis et~al.(2010)Dimakis, Kar, Moura, Rabbat, and
  Scaglione]{journals/pieee/DimakisKMRS10}
Alexandros~G. Dimakis, Soummya Kar, José M.~F. Moura, Michael~G. Rabbat, and
  Anna Scaglione.
\newblock Gossip algorithms for distributed signal processing.
\newblock \emph{Proceedings of the IEEE}, 98\penalty0 (11):\penalty0
  1847--1864, 2010.

\bibitem[Frieze and Grimmett(1985)]{DBLP:journals/dam/FriezeG85}
Alan~M. Frieze and Geoffrey~R. Grimmett.
\newblock The shortest-path problem for graphs with random arc-lengths.
\newblock \emph{Discrete Applied Mathematics}, 10\penalty0 (1):\penalty0
  57--77, 1985.

\bibitem[Pittel(1987)]{Pittel:1987:SR:37387.37400}
Boris Pittel.
\newblock On spreading a rumor.
\newblock \emph{SIAM J. Appl. Math.}, 47\penalty0 (1):\penalty0 213--223, 1987.

\bibitem[Eugster et~al.(2004)Eugster, Guerraoui, Kermarrec, and
  Massouli{\'{e}}]{EGMM2004}
Patrick~Th. Eugster, Rachid Guerraoui, Anne{-}Marie Kermarrec, and Laurent
  Massouli{\'{e}}.
\newblock Epidemic information dissemination in distributed systems.
\newblock \emph{{IEEE} Computer}, 37\penalty0 (5):\penalty0 60--67, 2004.

\bibitem[Koldehofe(2008)]{koldehofe2004simple}
Boris Koldehofe.
\newblock Simple gossiping with balls and bins.
\newblock \emph{Stud. Inform. Univ.}, 6\penalty0 (2):\penalty0 137--168, 2008.

\bibitem[Fraigniaud and Giakkoupis(2010)]{DBLP:conf/spaa/FraigniaudG10}
Pierre Fraigniaud and George Giakkoupis.
\newblock On the bit communication complexity of randomized rumor spreading.
\newblock In \emph{{SPAA} 2010: Proceedings of the 22nd Annual {ACM} Symposium
  on Parallelism in Algorithms and Architectures}, pages 134--143, 2010.

\bibitem[Avin and Els{\"{a}}sser(2013)]{avin2013faster}
Chen Avin and Robert Els{\"{a}}sser.
\newblock Faster rumor spreading: Breaking the logn barrier.
\newblock In \emph{Distributed Computing - 27th International Symposium,
  {DISC}}, pages 209--223, 2013.

\bibitem[Haeupler and Malkhi(2014)]{haeupler2014optimal}
Bernhard Haeupler and Dahlia Malkhi.
\newblock Optimal gossip with direct addressing.
\newblock In \emph{Proceedings of the 2014 ACM symposium on Principles of
  distributed computing, {PODC}}, pages 176--185. ACM, 2014.

\bibitem[Panagiotou et~al.(2013)Panagiotou, Pourmiri, and
  Sauerwald]{panagiotou2013faster}
Konstantinos Panagiotou, Ali Pourmiri, and Thomas Sauerwald.
\newblock Faster rumor spreading with multiple calls.
\newblock In \emph{Algorithms and Computation - 24th International Symposium,
  {ISAAC}}, pages 446--456, 2013.

\bibitem[Doerr and Fouz(2011)]{doerr2011asymptotically}
Benjamin Doerr and Mahmoud Fouz.
\newblock Asymptotically optimal randomized rumor spreading.
\newblock In \emph{International Colloquium on Automata, Languages, and
  Programming}, pages 502--513. Springer, 2011.

\bibitem[Doerr et~al.(2016)Doerr, Doerr, Moran, and Moran]{doerr2016simple}
Benjamin Doerr, Carola Doerr, Shay Moran, and Shlomo Moran.
\newblock Simple and optimal randomized fault-tolerant rumor spreading.
\newblock \emph{Distributed Computing}, 29\penalty0 (2):\penalty0 89--104,
  2016.

\bibitem[Gasieniec and Pelc(1996)]{Gasieniec96adaptivebroadcasting}
Leszek Gasieniec and Andrzej Pelc.
\newblock Adaptive broadcasting with faulty nodes.
\newblock \emph{Parallel Computing}, 22\penalty0 (6):\penalty0 903--912, 1996.

\bibitem[Diks and Pelc(2000)]{Diks2000}
Krzysztof Diks and Andrzej Pelc.
\newblock Optimal adaptive broadcasting with a bounded fraction of faulty
  nodes.
\newblock \emph{Algorithmica}, 28\penalty0 (1):\penalty0 37--50, 2000.

\bibitem[Alistarh et~al.(2010)Alistarh, Gilbert, Guerraoui, and
  Zadimoghaddam]{DBLP:conf/icalp/AlistarhGGZ10}
Dan Alistarh, Seth Gilbert, Rachid Guerraoui, and Morteza Zadimoghaddam.
\newblock How efficient can gossip be? (on the cost of resilient information
  exchange).
\newblock In \emph{Automata, Languages and Programming, 37th International
  Colloquium, {ICALP}}, pages 115--126, 2010.

\bibitem[Fountoulakis and Panagiotou(2010)]{fountoulakis2010rumor}
Nikolaos Fountoulakis and Konstantinos Panagiotou.
\newblock Rumor spreading on random regular graphs and expanders.
\newblock In \emph{Proceedings of {RANDOM} 2010}, volume 6302 of \emph{Lecture
  Notes in Computer Science}, pages 560--573, 2010.

\bibitem[Giakkoupis(2011)]{DBLP:conf/stacs/Giakkoupis11}
George Giakkoupis.
\newblock Tight bounds for rumor spreading in graphs of a given conductance.
\newblock In \emph{28th International Symposium on Theoretical Aspects of
  Computer Science, {STACS} 2011}, pages 57--68, 2011.

\bibitem[Gilbert et~al.(2017)Gilbert, Robinson, and
  Sourav]{DBLP:conf/podc/GilbertRS17}
Seth Gilbert, Peter Robinson, and Suman Sourav.
\newblock Brief announcement: Gossiping with latencies.
\newblock In \emph{Proceedings of the {ACM} Symposium on Principles of
  Distributed Computing, {PODC}}, pages 255--257, 2017.

\bibitem[Acan et~al.(2017)Acan, Collevecchio, Mehrabian, and
  Wormald]{acan2015push}
H{\"{u}}seyin Acan, Andrea Collevecchio, Abbas Mehrabian, and Nick Wormald.
\newblock On the push{\&}pull protocol for rumor spreading.
\newblock \emph{{SIAM} J. Discrete Math.}, 31\penalty0 (2):\penalty0 647--668,
  2017.

\bibitem[Mitzenmacher and Upfal(2005)]{mitzenmacher2005probability}
Michael Mitzenmacher and Eli Upfal.
\newblock \emph{Probability and computing - Randomized algorithms and
  probabilistic analysis}.
\newblock Cambridge University Press, 2005.

\bibitem[Mercier et~al.(2017)Mercier, Hayez, and
  Matos]{DBLP:conf/podc/MercierHM17}
Hugues Mercier, Laurent Hayez, and Miguel Matos.
\newblock Brief announcement: Optimal address-oblivious epidemic dissemination.
\newblock In \emph{Proceedings of the {ACM} Symposium on Principles of
  Distributed Computing, {PODC}}, pages 151--153, 2017.

\end{thebibliography}

\pagebreak

\end{document}